\documentclass[leqno,12pt,draft]{amsart}
\usepackage{bm}
\usepackage{amssymb}
\numberwithin{equation}{section}
\theoremstyle{plain} % italic
\newtheorem{theorem}{\indent\bf Theorem}[section] 
\newtheorem{lemma}[theorem]{\indent\bf Lemma}
\newtheorem{corollary}[theorem]{\indent\bf Corollary}

\theoremstyle{definition} % 

\newtheorem{remark}[theorem]{\indent\bf Remark}

\begin{document}
\begin{flushright}
\today
\end{flushright}
\title[Exponential polynomials and Quantum computing]
{On zeros of exponential polynomials \\and quantum algorithms} %

\author[Y. Sasaki]{Yoshitaka Sasaki$^*$}

%%%%%%%%%%%%%%%%%%% ‹r' %%%%%%%%%%%%%%%%%%%%%%%%%%%%%%
%\subjclass[2000]{ % 2000MSC
%Primary 00; Secondary 00.}

\thanks{$*$Partly supported by ``Open Research Center" Project 
for Private Universities: matching fund subsidy from MEXT}

%%%%%%%%%%%% affiliation %%%%%%%%%%%%%

\address{ 
Interdisciplinary Graduate School of Science and Engineering \endgraf
Kinki University \endgraf
Higashi-Osaka, Osaka 577-8502 \endgraf
Japan}
\email{sasaki@alice.math.kindai.ac.jp}

%%%%%%%%%%%%%%%%%%%%%%%%%%%%%%%%%%%%%%%%%%%%%%%%%%%%%%%

\maketitle
\begin{abstract}
We calculate the zeros of an exponential polynomial of some variables 
by a classical algorithm and quantum algorithms which are 
based on the method of van Dam and Shparlinski, they 
treated the case of two variables, 
and compare with the complexity of those cases. 
Further we consider the ratio (classical$/$quantum) of the complexity. 
Then we can observe the ratio is virtually $2$ when the number of the variables 
is sufficiently large. 
\end{abstract}

%%%%%%%%%%%%%%%%%%%%%%%%%%%%%%%%%%%%%%%%%%%%%%%%%%
\section{Introduction} \label{intro}
%%%%%%%%%%%%%%%%%%%%%%%%%%%%%%%%%%%%%%%%%%%%%%%%%%%%%%%%
For a prime number $p$, we put $q = p^{\nu}$, where $\nu$ is 
a certain positive integer. 
Then we denote the finite field by $\mathbb{F}_q$ which has $q$ elements. 
Namely, $\mathbb{F}_q$ forms an additive group and 
$\mathbb{F}_q^{\times} := \mathbb{F}_q \backslash \{ 0 \}$ forms a 
multiplicative group, where $0$ is the zero element in $\mathbb{F}_q$. 
Any element of $\alpha \in \mathbb{F}_q^{\times}$ have a periodicity, that is 
there exists a smallest natural number $s$ such that $\alpha^s =1$. We call 
such $s$ the ``multiplicative order" of $\alpha$. It is known that 
the multiplicative order is a divisor of $\# \mathbb{F}_q^{\times} = q-1$. 
See \cite{ln}, \cite{cd} for the details. 

To evaluate the number of zeros of a homogeneous polynomial 
\[ F (x_0, \dots, x_{m}) = \sum_{(n_0, \dots, n_{m}) \in \mathbb{N}_0^{m+1}} 
a_{n_0, \dots, n_{m}} x_0^{n_0} \cdots x_{m}^{n_{m}} \]
is a very important problem in mathematics. Here, 
$\mathbb{N}_0 := \mathbb{N} \cup \{ 0 \}$ and 
$a_{n_1, \dots, n_{m}} \in \mathbb{F}_q$. 
The zeta-function associated with such polynomial (the congruence zeta-function)
 was introduced to treat this problem. 
Particularly, the zeros of the congruence zeta-function satisfies an analogue of 
the Riemann hypothesis called ``Weil conjecture". Therefore to compute the zeros of 
the congruence zeta-function is very important investigation. 
In \cite{d}, van Dam studied the zeros of 
the zeta-function associated with the Fermat surface by using quantum computing. 

In \cite{ds}, van Dam and Shparlinski treated 
the following exponential polynomial 
\begin{equation}
f (x, y) = a_1 g_1^{x} + a_2 g_2^{y} -b \label{eq2}
\end{equation}
and calculated the zeros of \eqref{eq2} by quantum algorithms. 
Further they compared the complexity due to a classical algorithm with 
that due to a quantum algorithm. Then the ``cubic" speed-up was observed. 

In this article, we treat the exponential polynomial of $n$ variables 
\begin{equation}
f_b (x_1, \dots, x_n) := a_1 g_1^{x_1} + \cdots + a_n g_n^{x_n} -b. \label{eq}
\end{equation}
We restrict $n \ll q^{\varepsilon}$ with a small $\varepsilon > 0$. 
The reason why we claim this restriction will be explained in Appendix, below. 
We calculate the solutions of $f_b (x_1, \dots, x_n) = 0$ by 
using quantum algorithms which are natural generalizations of the method 
of van Dam and Shparlinski. 
Here, $a_i$, $g_i \in \mathbb{F}_q^{\times}$ ($i =1, \dots, n$) and 
$b \in \mathbb{F}_q$. Further we also compare the complexity 
due to a classical algorithm with 
that due to a quantum algorithm. 
Then exponentially ``$(2n-1)/(n-1)$" times speed-up is observed. 
We notice that $(2n-1)/(n-1) = 2+1/(n-1)$ is virtually $2$ when $n$ 
is sufficiently large. 
This is the boundary between a standard classical algorithm and our quantum 
algorithm. 
In the previous paper \cite{osy}, Ohno, the author and Yamazaki treated 
the case of three variables and obtained the ratio $5/2$. 

In the next section, we introduce some notation and give the considerable 
lemma which supports whether there exist the zeros of \eqref{eq}. 
In Section \ref{cl}, we evaluate the complexity due to a classical algorithm. 
Further in Section \ref{qu}, we evaluate the complexity due to 
a quantum algorithm. 

%%%%%%%%%%%%%%%%%%%%%%%%%%%%%%%%%%%%%%%%%%%%%%%%%%%%%%%%%%%%
\section{The number of solution of equation}
%%%%%%%%%%%%%%%%%%%%%%%%%%%%%%%%%%%%%%%%%%%%%%%%%%%%%%%%%%%%%%
In this section, we give an important formula with respect to 
the density of solutions of 
\begin{equation}
f_b (x_1, \dots, x_n) = 0
\label{th1-eq}
\end{equation}
as Lemma \ref{lem1}, below. 
To state it, we introduce some notation. 

Let each $s_i$ be the multiplicative order of 
$g_i$ ($i=1, \dots, n$) in \eqref{th1-eq}. 
We put 
\begin{align*}
X_i &:= \{ 0, 1, \dots, s_i-1 \} \cong \mathbb{Z}/s_i \mathbb{Z}, \quad 
\text{($i=1, \dots, n$),} \\
X_{n} (r) &:= \{ 0, 1, \dots, r-1 \} \subseteq X_n \quad 
\text{$(r = 1,2, \dots, s_n)$,} \\
\bm{X}^n (r) &:= X_1 \times \cdots \times X_{n-1} \times X_{n} (r), \\ 
\bm{X}^n &:= \bm{X}^n (s_n) = X_1 \times \cdots \times X_{n-1} \times X_{n} 
\intertext{and}
\vec{x} &:= (x_1, \dots, x_n) \in \bm{X}^n (r). 
\end{align*}
Then we define
\begin{align*}
S_{f_b} (r) &:= \{ (x_1, \dots, x_n) \in \bm{X}^n (r) \ | 
\ f_b (x_1, \dots, x_n) = 0 \}, \\
N_{f_b} (r) &:= \# S_{f_b} (r) 
\end{align*}
for $r = 1, \dots, s_n$. 

By using above notation, we can state the following result:

\begin{lemma} \label{lem1}
Let $\delta$ be a parameter satisfying $\delta = o(q)$. 
For $r > \delta^2 q^n (\prod_{l=1}^{n-1} s_l)^{-2}$, 
we have 
\begin{equation}
N_{f_b} (r) = \frac{r \prod_{l=1}^{n-1} s_{l}}{q} 
+ O (\delta \sqrt{r q^{n-2}}), 
\end{equation}
except for at most $q/\delta^2$ exceptional $b$'s. 
Further $O$-constant can be taken $1$. 
\end{lemma}

Choosing $\delta = (\log q)^{1/2}$ in Lemma \ref{lem1}, we have 
\begin{corollary} \label{cor}
If $q^n (\prod_{l=1}^{n-1} s_l)^{-2} \log q < r \leq s_n$, 
then we see that $S_{f_b} (r) \neq \phi$ holds except for 
at most $q/\log q$ exceptional $b$'s. 
\end{corollary}

\begin{remark}
The exponent $1/2$ of $\delta = (\log q)^{1/2}$ is not necessary. 
In fact, $\delta = (\log q)^{\varepsilon}$ 
with any $\varepsilon > 0$ is sufficient. 
\end{remark}

\begin{proof}[\indent\sc Proof of Lemma \ref{lem1}]
Let $\psi$ be a non-trivial additive character over $\mathbb{F}_q$, in fact, any additive 
character over $\mathbb{F}_q$ can be given as a map $\mathbb{F}_q \to \mathbb{C}_1^*$, where 
$\mathbb{C}_1^* := \{ z \in \mathbb{C} | |z| = 1 \}$ (see \cite[Theorem 5.7]{ln}). 
To evaluate $N_{f_b} (\bm{v})$, we use the following formula which plays as a counting function:
\begin{equation}
\frac{1}{q} \sum_{\mu \in \mathbb{F}_q} \psi (u\mu) = 
\begin{cases}
1 & \text{if $u = 0$,} \\
0 & \text{otherwise.}
\end{cases} \label{or}
\end{equation}
Then we have
\begin{align}
N_{f_b} (r) &= \sum_{\vec{x} \in \bm{X}^n (r)}
\frac{1}{q} \sum_{\mu \in \mathbb{F}_q} 
\psi (\mu (f_b (x_1, \dots, x_n))) \label{nf} \\
&= \frac{r \prod_{j=1}^{n-1} s_l}{q} + \frac{1}{q} \sum_{\mu \in \mathbb{F}_q^*} 
\sum_{\vec{x} \in \bm{X}^n (r)} 
\psi (\mu (f_b (x_1, \dots, x_n))) \notag \\
&=: \frac{r \prod_{l=1}^{n-1} s_l}{q} + \Delta_b (r). \notag
\end{align}
If the contribution from the second term on the right-hand side of 
the above formula 
can be estimated by $o(r \prod_{l=1}^{n-1} s_l/q)$, the above formula tells us 
the existence of the solution of $f_b (x_1, \dots, x_n)$. 
To consider it, we evaluate the mean value of the second term on the right-hand 
side of \eqref{nf} 
with respect to $b$. Namely, we evaluate 
\[ E (r) := \sum_{b \in \mathbb{F}_q} 
\left| \Delta_b (r) \right|^2. \]
From \eqref{or} and some properties of the additive character over $\mathbb{F}_q$, 
we obtain
\begin{align*}
E (r) =& \frac{1}{q^2} \sum_{\mu, \mu' \in \mathbb{F}_q^{\times}} 
\left( \prod_{j=1}^{n-1} \left( \sum_{x_j, x_j' \in X_j} 
\psi (a_j (\mu g_j^{x_j} - \mu' g_j^{x_j'})) \right) \right)
\sum_{x_n, x_n' \in X_n (r)} 
\psi (a_n (\mu g_n^{x_n} - \mu' g_n^{x_n'})) \\
& \quad \times \sum_{b \in \mathbb{F}_q} \psi (b(\mu'-\mu)) \\
=& \frac{1}{q} \sum_{\mu \in \mathbb{F}_q^{\times}} 
\left( \prod_{j=1}^{n-1} \left( \sum_{x_j, x_j' \in X_j} 
\psi (a_j \mu (g_j^{x_j} - g_j^{x_j'})) \right) \right) 
\sum_{x_n, x_n' \in X_n (r)} 
\psi (a_n \mu (g_n^{x_n} - g_n^{x_n'})) \\
=& \frac{1}{q} \sum_{\mu \in \mathbb{F}_q^{\times}} 
\left( \prod_{j=1}^{n-1} \Biggl| \sum_{x_j \in X_j} 
\psi (a_j \mu g_j^{x_j}) \Biggr|^2 \right)
\Biggl| \sum_{x_n \in X_n (r)} \psi (a_n \mu g_n^{x_n}) \Biggr|^2. 
\end{align*}
It is known that 
\begin{align*}
\Biggl| \sum_{x_j \in X_j} \psi (a_j \mu g_j^{x_j}) \Biggr| 
&\leq \sqrt{q} \quad \text{for $j=1,\dots, n-1$ and 
any $\mu \in \mathbb{F}_q^{\times}$} 
\end{align*}
(see Theorem 8.78 in \cite{ln}). Hence we have
\begin{align*}
E (r) <& q^{n-2}
\sum_{\mu \in \mathbb{F}_q} 
\Biggl| \sum_{x_n \in X_n (r)} \psi (a_n \mu f^{x_n}) \Biggr|^2 
= q^{n-1} r. 
\end{align*}
Therefore, if we put $\delta = o(q)$, then we can see that 
there exist at most $q/\delta^2$ exceptional $b$'s such that 
\begin{equation}
\left| \frac{1}{q} \sum_{\mu \in \mathbb{F}_q^*} 
\sum_{\vec{x} \in X_n (r)} 
\psi (\mu (f_b (x_1, \dots, x_n))) \right| \geq \delta \sqrt{rq^{n-2}}. 
\end{equation}
Hence we obtain 
\begin{equation*}
N_{f_b} (r) = \frac{r \prod_{l=1}^{n-1} s_l}{q} 
+ O ( \delta \sqrt{q^{n-2} r} )
\end{equation*}
for other $b$'s. 
Now, the proof of Lemma \ref{lem1} is completed. 
\end{proof}

%%%%%%%%%%%%%%%%%%%%%%%%%%%%%%%%%%%%%%%%%%%%%%%%%%%
\section{Calculation of the deterministic time for a classical algorithm} \label{cl}
%%%%%%%%%%%%%%%%%%%%%%%%%%%%%%%%%%%%%%%%%%%%%%%%%%%%%%
We follow the method of van Dam and Shparlinski~\cite{ds}. Then we have 
\begin{theorem} \label{th1}
Except for at most $q/\log q$ exceptional $b$'s, 
we can either find a solution $\vec{x} \in \bm{X}^n$ of the equation  
\eqref{th1-eq} or decide that it does not have a solution in deterministic time 
$q^{n(n+1)/2(2n-1)} (\log q)^{O(1)}$ as a classical computer. 
\end{theorem}

\begin{proof}
Using a standard deterministic factorization algorithm, we factorize $q-1$ 
and find the orders $s_j$ of $g_j$ ($j=1,\dots, n$) 
in time $q^{1/2} (\log q)^{O(1)}$. 
We may assume without loss of generality that $s_1 \geq \cdots \geq s_n$. 
For calculated orders $s_1, \dots, s_{n-1}$, we put 
\begin{equation}
r = \Bigl\lceil q^n \Bigl( \prod_{l=1}^{n-1} s_l \Bigr)^{-2} \log q 
\Bigr\rceil. \label{r}
\end{equation}
Then we see that the solution of \eqref{th1-eq} certainly exists 
when $r \leq s_n$. However, when $r > s_n$, 
we do not know whether such solutions exist. Therefore we have to 
consider those two cases. 

For each $(x_2, \dots, x_{n-1}, x_n) \in X_2 \times \cdots 
\times X_{n-1} \times X_n (r)$, 
we calculate the deterministic time of the discrete logarithm $x_1$ such that 
$g_1^{x_1} = a_1^{-1} ( b-a_2 g_2^{x_2} -\cdots - a_n g_n^{x_n})$. 
It is known that the deterministic time for this case is 
$s_1^{1/2} (\log q)^{O(1)}$ (see Section 5.3 in \cite{cp}). 
\begin{enumerate}
\item The case $r \leq s_n$. 
We have 
\[ s_1^{1/2} \Bigl( \prod_{l=2}^{n-1} s_l \Bigr) r (\log q)^{O(1)} \ll 
q^{n/2} (\log q)^{O(1)}, \]
since $s_1^{1/2} (\prod_{l=2}^{n-1} s_l) r < ((\prod_{l=1}^{n-1} s_l)^2 r)^{1/2}$. 
\item The case $r > s_n$. Similarly, we see that the deterministic time is 
\[ s_1^{1/2} \Bigl( \prod_{l=2}^{n} s_l \Bigr) (\log q)^{O(1)} 
\ll q^{n/2} (\log q)^{O(1)}, \]
since $s_1^{1/2} \prod_{l=2}^n s_l < ((\prod_{l=1}^{n-1} s_{l})^2 s_n)^{1/2} 
< ((\prod_{l=1}^{n-1} s_l)^2 r)^{1/2}$. 
\end{enumerate}
\end{proof}

%%%%%%%%%%%%%%%%%%%%%%%%%%%%%%%%%%%%%%%%%%%%%%%%%%%%%%%%%%%%%%%%%%%%%%%%
\section{Calculation of the complexity for a quantum algorithm} \label{qu}
%%%%%%%%%%%%%%%%%%%%%%%%%%%%%%%%%%%%%%%%%%%%%%%%%%%%%%%%%%%%%%%%%%%%%%%%
In this section, we describe quantum algorithms which are based on the 
method of \cite{ds}. Hereafter $\varepsilon$ is any positive and small real 
number. 

\begin{theorem} \label{th2}
Except for at most $q/\log q$ exceptional $b$'s, 
we can either find a solution $\vec{x} \in \bm{X}^n$ of the equation 
\eqref{th1-eq} 
or decide that it does not have a solution in time 
$q^{n(n-1)/2(2n-1)+\varepsilon} (\log q)^{O(1)}$ as a quantum computer. 
\end{theorem}

\begin{proof}
Using Shor's algorithm~\cite{sh}, we can obtain the multiplicative orders 
$s_j$'s ($j=1, \dots, n$) in polynomial time. 
We may assume without loss of generality that $s_1 \geq \cdots \geq s_n$. 
As in the proof of Theorem \ref{th1}, we put $r$ as \eqref{r}. 
Further, we consider a polynomial time quantum subroutine 
$\mathcal{S} (x_2, \dots, x_n)$ which either finds and returns $x_1 \in X_1$ with 
\[ g_1^{x_1} = a_1^{-1} ( b-a_2 g_2^{x_2} -\cdots -a_n g_n^{x_n}) \]
or reports that no such $x_1$ exists for a given 
$(x_2, \dots, x_{n-1}, x_n) \in X_2 \times\cdots \times X_{n-1} \times X_n (r)$ 
by using Shor's discrete logarithm algorithm. 
\begin{enumerate}
\item The case $r \leq s_n$. 
Using Grover's search algorithm~\cite{g}, we search the subroutine 
$\mathcal{S} (x_2, \dots, x_n)$ for all $(x_2, \dots, x_{n-1}, x_n) 
\in X_2 \times \cdots \times X_{n-1} \times X_n (r)$ in time 
\[ q^{\varepsilon} (r \prod_{l=2}^{n-1} s_l)^{1/2} (\log q)^{O(1)} \ll 
q^{n(n-1)/2(2n-1)+\varepsilon} (\log q)^{O(1)}, \]
since $r \prod_{l=2} s_l \leq ((\prod_{l=1}^{n-1} s_l)^2 r)^{(n-1)/(2n-1)}$. 
\item The case $r > s_n$. 
Similarly, we search the $\mathcal{S} (x_2, \dots, x_n)$ for all 
$(x_2, \dots, x_{n-1}, x_n) \in X_2 \times \cdots \times X_{n-1} \times X_n (r)$ 
in time
\[ q^{\varepsilon} (\prod_{l=2}^n s_l)^{1/2} (\log q)^{O(1)} 
\ll q^{n(n-1)/2(2n-1)+\varepsilon} (\log q)^{O(1)}, \]
since $\prod_{l=2}^n s_n \leq ((\prod_{l=1}^{n-1} s_l)^2 s_n)^{(n-1)/(2n-1)} 
< ((\prod_{l=1}^{n-1} s_l)^2 r)^{(n-1)/(2n-1)}$. 
\end{enumerate}
\end{proof}

In \cite{ds}, van Dam and Shparlinski mentioned when the multiplicative orders
are large, there is a more efficient quantum algorithm. Similarly, we can 
also consider a more efficient quantum algorithm. 

\begin{theorem} \label{th3}
If we assume 
\[ \Bigl( \prod_{l=1}^{n-1} s_l \Bigr)^2 s_n > q^n \log q, \]
then we can either find a solution $\vec{x} \in \bm{X}^n$ of the equation 
\eqref{th1-eq} or decide that it does not have a solution in time 
$q^{1/2+\varepsilon} ((\prod_{l=1}^{n-1} s_l)^2 s_n)^{-1/2(2n-1)} (\log q)^{O(1)}$ 
as a quantum computer, 
except for at most $q/\log q$ exceptional $b$'s. 
\end{theorem}

\begin{remark}
The upper bound of the running time of the algorithm of Theorem \ref{th3} is 
\[ O(q^{(n-1)/2(2n-1)+\varepsilon} (\log q)^{O(1)}). \]
\end{remark}

\begin{proof}[\indent\sc Proof of Theorem \ref{th3}]
We may assume without loss of generality that $s_1 \geq s_2 \geq s_3$. 
We put 
\begin{equation}
r = \Bigl\lfloor q^n \Bigl( \prod_{l=1}^{n-1} s_l \Bigr)^{-2} \log q \Big\rfloor
\end{equation}
Then from the assumption of the theorem we see that $r \leq s_n$. Hence 
there are some solutions of \eqref{th1-eq} in $\bm{X}^n (r)$ and 
we denote the number of the solutions of \eqref{th1-eq} by $M$. 
Note that $M \asymp (r \prod_{l=1}^{n-1} s_l)/q$. 

As in the case of \cite{ds}, we use the version of Grover's algorithm as 
described in \cite{bbht} that finds one out of $m$ matching items in a set 
of size $t$ by using only $O(\sqrt{t/m})$ queries. We search the subroutine 
$\mathcal{S} (x_2, \dots, x_n)$ for all $(x_2, \dots, x_{n-1}, x_n) \in 
X_2 \times \cdots \times X_{n-1} \times X_n (r)$. Then the complexity is 
\[ q^{\varepsilon} \Bigl( \frac{(\prod_{l=2}^{n-1} s_l) r}{M} \Bigr)^{1/2} 
(\log q)^{O(1)} 
\leq q^{1/2+\varepsilon} \Bigl( \Bigl( \prod_{l=1}^{n-1} s_l 
\Bigr)^2 s_n \Bigr)^{-1/2(2n-1)} 
(\log q)^{O(1)}. \]
\end{proof}

%%%%%%%%%%%%%%%%%%%%%%%%%%%%%%%%%%%%%%%
\section{Concluding remarks}
%%%%%%%%%%%%%%%%%%%%%%%%%%%%%%%%%%%%%%%%%%%
See the following list. 
\begin{center}
\begin{tabular}{|c|c|c|c|} \hline
$\#$ of variables & Classical & Quantum & ratio (C/Q) \\ \hline
2 \ (van Dam and Shparlinski) & 1 & 1/3 & 3 \\ \hline
3 \ (Ohno, S, Yamazaki) & 3/2 & 3/5 & 5/2 \\ \hline
\vdots & \vdots & \vdots & \vdots \\ \hline
$n$ & $n/2$ & $n(n-1)/2(2n-1)$ & $(2n-1)/(n-1)$ \\ \hline
\end{tabular}
\end{center}
We notice that the ratio is virtually $2$ when $n$ is sufficiently large. 
It seems to come from the effect of Grover's algorithm.

\appendix
%%%%%%%%%%%%%%%%%%%%%%%
\section{Appendix}
%%%%%%%%%%%%%%%%%%%%%%%%%%
For a natural number $n$, we define the divisor function by 
\begin{equation}
d (n) := \# \{ d \in \mathbb{N} \ | \ d|n \} = \sum_{d|n} 1. 
\end{equation}
In Section \ref{intro}, we introduced the notion of the multiplicative order $s$ 
of $g \in \mathbb{F}_q^{\times}$ and mentioned the multiplicative order 
is a divisor of $q-1$. 

We put $g_1$ and $g_2$ have the same multiplicative order. Then there exists a 
natural number $l$ such that $g_2 = g_1^l$. Hence, we have 
\begin{equation}
a_1 g_1^{x_1} + a_2 g_2^{x_2} = a_1 g_1^{x_1} + a_2 g_1^{lx_2}, 
\end{equation}
where $a_i, g_i \in \mathbb{F}_q^{\times}$ ($i=1,2$). The right-hand side of the 
above equation is a element of $\overline{a_1} + \overline{a_2} \in 
\mathbb{F}_q^{\times} / \langle g_1 \rangle$, 
where $\overline{a}$ is a coset of $\mathbb{F}_q^{\times} / \langle g_1 \rangle$ 
and $\langle g_1 \rangle$ is the cyclic group generated by $g_1$. Therefore 
our central problem \eqref{eq} is reduced to 
\begin{equation}
\widetilde{f}_b (z_1, \dots, z_{\mu}) 
= c_1 h_1^{z_1} + \cdots + c_{\mu} h_{\mu}^{z_{\mu}} -b = 0, 
\end{equation}
where each $h_i$ ($i=1, \dots, \mu$) does not have the same multiplicative order 
and $\mu \leq d(q-1)$.  

It is known that 
\[ d(n) \ll n^{\varepsilon} \] 
for any positive number $\varepsilon$ (for instance, see \cite{iv}). 
Hence, we have 
\[ \mu \ll q^{\varepsilon} \]
for any $\varepsilon > 0$.

\end{document}